\documentclass[runningheads]{llncs}
\usepackage{graphicx}
\begin{document}
\title{Backgammon is Hard}
\titlerunning{Backgammon is Hard}
\author{R. Teal Witter\inst{1}\orcidID{0000-0003-3096-3767}}
\authorrunning{R.T. Witter}
\institute{NYU Tandon, Brooklyn NY 11201, USA
\email{rtealwitter@nyu.edu}}
\maketitle
\begin{abstract}
We study the computational complexity of the popular board game backgammon. We show that deciding whether a player can win from a given board configuration is NP-Hard, PSPACE-Hard, and EXPTIME-Hard under different settings of known and unknown opponents' strategies and dice rolls. Our work answers an open question posed by Erik Demaine in 2001. In particular, for the real life setting where the opponent's strategy and dice rolls are unknown, we prove that determining whether a player can win is EXPTIME-Hard. Interestingly, it is not clear what complexity class strictly contains each problem we consider because backgammon games can theoretically continue indefinitely as a result of the capture rule.
\end{abstract}
\section{Introduction}
Backgammon is a popular board game played by two players.
Each player has 15 pieces that lie on 24
points evenly spaced on a board.
The pieces move in opposing directions according
to the rolls of two dice.
A player wins if they are the first to move all of their
pieces to their home and then off the board.

The quantitative study of backgammon began
in the early 1970's and algorithms for the game progressed quickly.
By 1979, a computer program had beat the World Backgammon
Champion 7 to 1 \cite{berliner1980backgammon}.
This event marked the first time a computer program
bested a reigning human player in a recognized
intellectual activity.
Since then, advances in backgammon programs
continue especially through the use of
neural networks
\cite{pollack1998co,tesauro1994td,tesauro2002programming}.

On the theoretical side, backgammon has been
studied from a probabilistic perspective as
a continuous process and random walk
\cite{keeler1975optimal,thorp2007backgammon}.
However, the computational complexity of backgammon
remains an open problem two decades after it was
first posed \cite{demaine2001playing}.
One possible explanation (given in online resources)
is that the generalization of backgammon is unclear.

\begin{table}[h]
    \centering
    \caption{A selection of popular games and computational
    complexity results.}
    \begin{tabular}{|c|c|c|}
    \hline
    Game& Complexity Class  \\ 
    \hline
    Tic-Tac-Toe
    & PSPACE-Complete 
    \cite{reisch1981hex} \\
    Checkers
    & EXPTIME-Complete
    \cite{robson1984n} \\
    Chess
    & EXPTIME-Complete \cite{fraenkel1981computing} \\
    Bejeweled
    & NP-Hard \cite{guala2014bejeweled} \\
    Go
    & EXPTIME-Complete \cite{robson1983complexity} \\
    Hanabi
    & NP-Hard \cite{baffier2017hanabi} \\
    Mario Kart
    & PSPACE-Complete \cite{bosboom2015mario} \\
    \hline
    \end{tabular}
    \label{tab:games}
\end{table}

From a complexity standpoint, backgammon stands in 
glaring contrast to many other popular games.
Researchers have established the complexity
of numerous games including those listed in
Table \ref{tab:games} but we are
not aware of any work on the complexity
of backgammon.

In this paper, we study the computational complexity of
backgammon.
In order to discuss the complexity of the game,
we propose a natural generalization of backgammon.
Inevitably, though, we have to make arbitrary choices such
as the number or size of dice in the generalized game.
Nonetheless, we make every
effort to structure our reductions
so that they apply to as many generalizations
as possible.

There are two main technical issues that make
backgammon particularly challenging to analyze.
The first is the difficulty in forcing a player
into a specific move.
All backgammon pieces follow 
the same rules of movement
and there are at least 15 unique combinations
of dice rolls
(possibly more for different generalizations)
per turn.
For other games, this problem has been solved by
more complicated reductions and extensive reasoning
about why a player has to follow specified moves
\cite{buchin2021dots}.
In our work, we frame the backgammon problem
from the perspective of a single player
and use the opponent and dice rolls
to force the player into predetermined moves.

The second challenge is that the backgammon
board is one-dimensional.
Most other games with 
computational complexity results have
at least two dimensions of
play which creates more structure in the
reductions \cite{hearn2009games}.
We avoid using multiple dimensions
by carefully picking Boolean satisfiability
problems to reduce from.

We show that deciding whether a player
can win is NP-Hard, PSPACE-Hard, and EXPTIME-Hard
for different settings of
known or unknown dice rolls and opponent
strategies.
In particular, in the setting most similar
to the way backgammon is actually played
where the opponent's strategy and dice rolls
are unknown,
we show that deciding whether
a player can win is EXPTIME-Hard.
Our work answers an open problem posed by Demaine
in 2001 \cite{demaine2001playing}.

In Section \ref{sec:gen}, we introduce the relevant
rules of backgammon and generalize it
from a finite board
to a board of arbitrary dimension.
In Section \ref{sec:np}, we prove that
deciding whether a player can 
win even when all future
dice rolls and the opponent's strategy are known
is NP-Hard.
In Section \ref{sec:pspace}, we prove that deciding
whether a player can win when 
dice rolls are known and the
opponent's strategy is unknown is PSPACE-Hard.
Finally in Section \ref{sec:exptime}, we prove that
deciding whether a player can win 
when dice rolls and
the opponent's strategy are unknown is
EXPTIME-Hard.

\section{Backgammon and its Generalization}\label{sec:gen}
We begin by describing the rules of backgammon
relevant to our reductions.
When played in practice, the backgammon board consists
of 24 points where 12 points lie on Player 1's side
and 12 points lie on Player 2's side.
However, without modifying the structure of the game,
we will think of the board as a line
of 24 points where Player 1's home consists of the rightmost
six points and Player 2's home consists of the leftmost
six points.
Figure \ref{fig:board} shows the relationship
between the regular board and our equivalent model.
Player 1 moves pieces right according to dice rolls
while Player 2 moves pieces left.
The goal is to move all of one's pieces home and then
off the board.
\begin{figure}
    \centering
    \includegraphics[scale=.18, page=2]{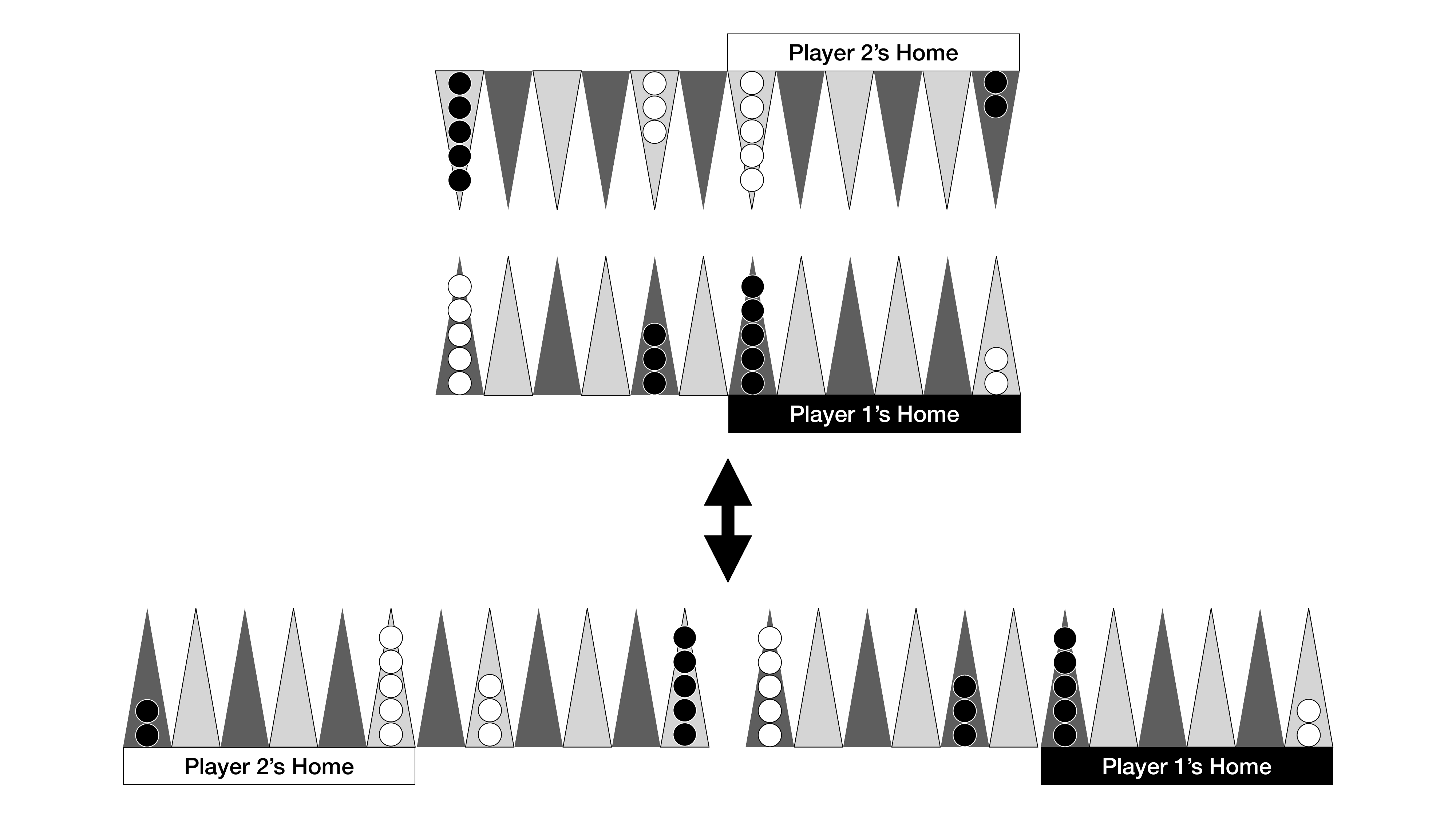}
    \caption{Backgammon board in normal play (top);
    equivalent board `unfolded' (bottom).}
    \label{fig:board}
\end{figure}

Players move their pieces by taking turns rolling
dice.
On their roll, a player may move one or more pieces
`forward' (right for Player 1 and left for Player 2)
by the numbers on the dice provided that
the new points are not blocked.
A point is blocked if the opponent has at least two
pieces on it.
The turn ends when either the player has moved their
pieces or all moves are blocked.
Note that a player must always use
as many dice rolls as possible
and if the same number appears on two dice then
the roll `doubles' so a player now has four moves
(rather than two) of the number.

If only one of a player's pieces is on a point,
the opponent may capture it by moving a piece to the point.
The captured piece is moved off the board and must
be rolled in from the opponent's home \textit{before}
any other move may be made.
This sets back the piece and can prove particularly
disadvantageous if all of the points in the opponent's
home are blocked.

The obvious way to generalize the backgammon board
used in practice is to concatenate multiple boards together,
keeping the top right as Player 1's home and the
bottom right as Player 2's home.
In the line interpretation,
we can equivalently view this procedure as adding
more points between the respective homes.
The rules we described above naturally extend.
We formalize this generalization in
Definition \ref{def:general}.

\begin{definition}[Generalized Backgammon]
    \label{def:general}
    Let $m$ be a positive integer given as input.
    We define constants $h \geq 6$,
    and $d \geq 2$ $s \geq 6$
    where the lower bounds originate
    from traditional backgammon.
    Then a generalized backgammon instance
    consists of $n$ points on a line with the leftmost $h$
    and rightmost $h$ marked as each player's home
    and $d$ dice each with $s$ sides.
    We require the number of pieces $p$
    to be polynomial in $m$ and specify our choice of
    $p$ in the reductions.
\end{definition}

In our proofs, we fix the constants without loss
of generality by modifying our reductions.
We assume $h=6$ home points
by blocking additional points in the opponent's home.
We also assume
$d=2$ dice and $s=6$ sides by
rolling blocked pieces for the player and using dummy moves
for the opponent.

\section{NP-Hardness}\label{sec:np}
In this section, we show that determining
whether a player can win against a
known opponent's strategy and known dice rolls (KSKR)
is NP-Hard.
We begin with formal definitions
of Backgammon KSKR and the NP-Complete
problem we reduce from.

\begin{definition}[Backgammon KSKR]
The input is a configuration on a generalized
backgammon board, a complete description
of the opponent's strategy, and all future
dice rolls both for the player and opponent.
The problem is to determine whether a player
can win the backgammon game from the backgammon
board against the opponent's strategy and with
the specified dice rolls.
\end{definition}

We do not require that the configuration
is easily reachable from the start state.
However, one can imagine that given sufficient
time and collaboration between two players,
any configuration is reachable using the
capture rule.

An opponent's strategy is known if the player knows
the moves the opponent will make from all
possible positions in the resulting game.
Notice that such a description can be very large.
However, in our reduction,
we limit the number of possible positions
by forcing the player to make specific moves and 
predetermining the dice rolls.
Therefore the reduction stays polynomial
in the size of the 3SAT instance.
We formalize this intuition in
Lemma \ref{lemma:poly}.

\begin{definition}[3SAT]
The input is a Boolean expression in
Conjunctive Normal Form (CNF)
where each clause has at most three variables.
The problem is to determine whether a
satisfying assignment to the CNF exists.
\end{definition}

Given any 3SAT instance,
we construct a backgammon board configuration,
an opponent strategy, and dice rolls so
that the solution to Backgammon KSKR
yields the solution to 3SAT.
Since 3SAT is NP-Complete \cite{karp1972reducibility},
Backgammon KSKR must be NP-Hard.
We state the result formally below.

\begin{theorem}\label{thm:np}
Backgammon KSKR is NP-Hard.
\end{theorem}

\begin{proof}
Our proof consists of a reduction from 3SAT
to Backgammon KSKR.
Assume we are given an arbitrary 3SAT instance
with $n$ variables and $k$ clauses.
First, we force Player 1 (black) to choose
either $x_i$ or $\neg x_i$ for every $i \in [n]$.
Then, we propagate their choice into the appropriate clauses
in the Boolean expression from the 3SAT instance.
Finally, we reach a board configuration where Player 1
wins if and only if they have
chosen an assignment of bits that 
satisfies the Boolean expression.

We compartmentalize the process into ``gadgets.''
Each gadget simulates the behavior
of a part of the 3SAT problem:
There is an assignment gadget for every variable
that forces Player 1 to choose
either $x_i$ or $\neg x_i$ for every $i \in [n]$.
There is a clause gadget for every clause
that records whether the assignment Player 1
chose satisfies $c_j$ for every $j \in [k]$.

Player 1 wins if and only if their assignment
satisfies all clauses.
We ensure this by putting a single black piece
in each clause.
Player 1 satisfies the clause by protecting
their piece.
Once the assignment has been propagated
to the clauses,
Player 2 (white) captures any unprotected piece.
If even a single clause is unsatisfied (i.e. a
single piece is open), Player 2 traps the piece
and moves all the white pieces home before Player 1
can make a single additional move.
We block Player 2's home and use the rule
that a captured piece must be rolled in the
board before any other move can be made.

If, on the other hand, Player 1 satisfies every
clause then Player 1 will win since we will feed
rolls with larger numbers to Player 1 and smaller
numbers to Player 2.
Player 1 will then beat Player 2 given their
material advantage in the number 
of pieces on the board.

\begin{figure}
    \centering
    \includegraphics[scale=.18, page=7]{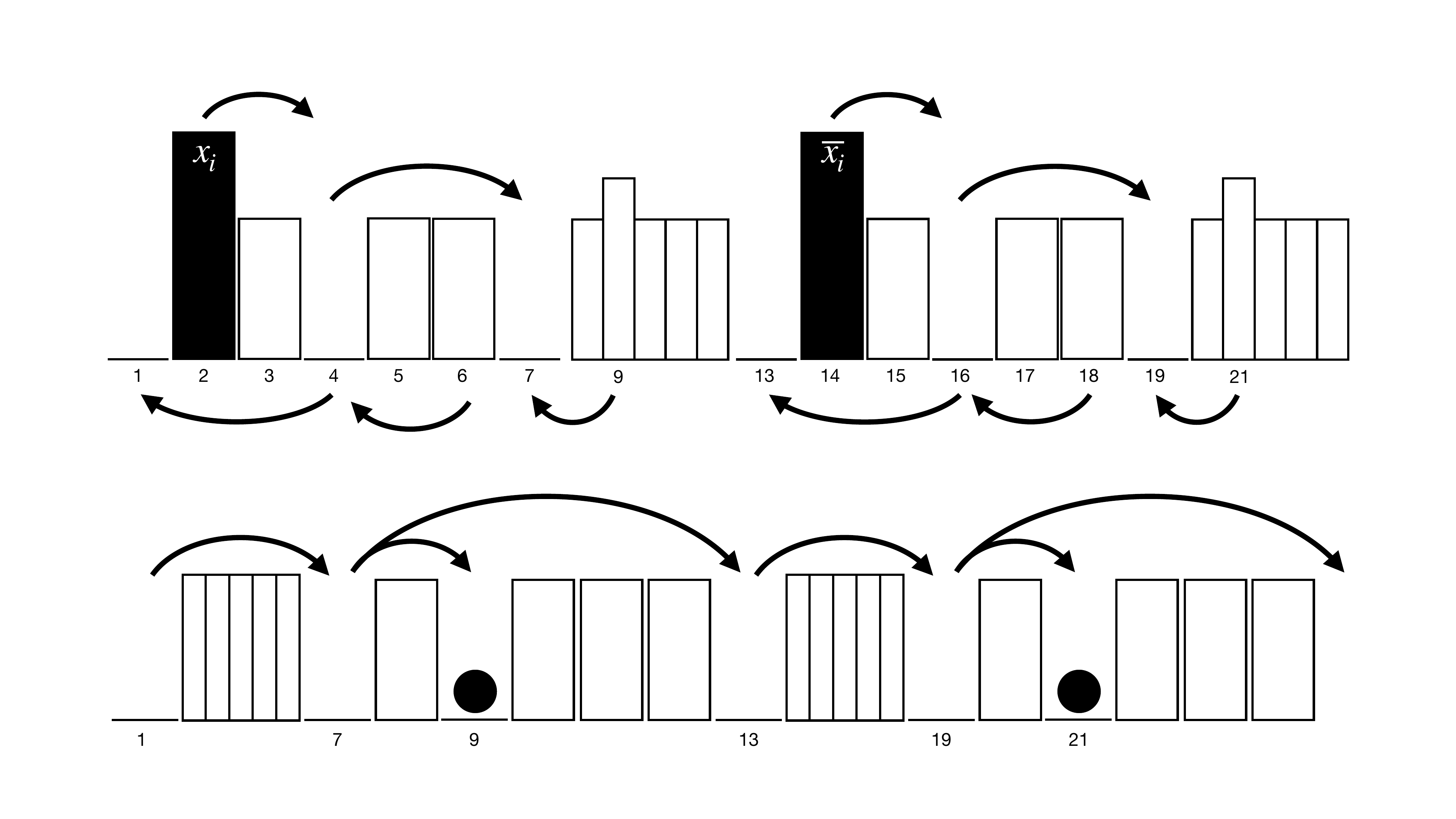}
    \caption{Reduction from 3SAT: variable gadget (top) for setting $x_i$ and two clause gadgets (bottom).}
    \label{fig:nphard}
\end{figure}

We now describe the variable and clause gadgets
in Figure \ref{fig:nphard}.
In order to simplify the concepts,
we explain the gadgets in the context of their
function in the reduction rather than
providing a complicated, technical definition.
There are $n$ variable gadgets
followed by $k$ clause gadgets arrayed in
increasing order of index from left to right.

For each $x_i$, we repeat the following process:
There are initially two white pieces each on point 4
and point 16
(the top of Figure \ref{fig:nphard}).
We move these pieces to 1 and 13 respectively while
feeding Player 1 blocked rolls e.g. one.
We then give Player 1 a two and a three.
The only moves they can make are from 2 to
4 to 7 or from 14 to 16 to 19.
This choice corresponds to setting $x_i$.
Without loss of generality, Player 1 chooses
$\overline{x_i}$ and Player 2 blocks 4 from 6
and 7 from 9.
We feed Player 1 rolls of two and three until all the
$\overline{x_i}$ pieces are on 19;
the number of these pieces is exactly the number
of times $\overline{x_i}$ appears in clauses.
We give Player 1 enough rolls to move all the
pieces corresponding to either $x_i$ or $\overline{x_i}$.
We then move to the next variable.

Once all the variables have been set, we move
down the variable gadgets from $x_n$
to $x_1$ propagating the choice of $x_i$ to
the appropriate clauses.
We use Player 2's pieces to block 1 and 13
for all variable gadgets with lower indices
so only the pieces in gadget $x_i$ can move.
(Variable gadgets with higher indices 
have already been emptied to clauses.)
For each $x_i$, we move the pieces on 19 through
the variable gadgets $x_{i+1},\ldots,x_{n}$ with
rolls of sixes.
Once we reach the clause gadgets, we move
one piece at a time with sixes
until we reach a clause that contains $\overline{x_i}$.
Once it reaches its clause,
we give the piece a two to protect the open piece.

We use a similar set of rolls for the $x_i$ pieces on 7.
Since the rolls have to be deterministic, we give
the rolls for both the $x_i$ and $\overline{x_i}$
before moving on to the $x_{i-1}$ variable gadget.
The rolls for whichever of $x_i$ and $\overline{x_i}$
Player 1 did not choose simply cannot be used.

Notice that every roll we give Player 1
can be played by exactly one piece
(except when Player 1 sets $x_i$).
While Player 1 receives rolls, we give
Player 2 `dummy' rolls of one and two to be
used on a stack of pieces near Player 2's home.

Once all variables have been set and the choices
propagated to the clauses, we give Player 2
a one to capture any unprotected pieces.
If Player 2 captures the unprotected piece,
Player 2's home is blocked so Player 1 cannot make
any additional moves until all of Player 2's pieces
are in their home.
At this point, Player 2 easily wins.
Otherwise, none of Player 1's pieces
are captured and the game becomes a race to the finish.
We give Player 2 low rolls and Player 1 high rolls
so Player 1 quickly advances and wins.

Since Player 1 wins if and only if
they find a satisfying assignment,
determining whether Player 1 can win
determines whether a satisfying
assignment to the 3SAT instance exists.
Then, with lemma \ref{lemma:poly},
Backgammon KSKR reduces from 3SAT in
polynomial space and time so Backgammon KSKR
is NP-Hard.
\end{proof}

\begin{lemma}\label{lemma:poly}
    The reduction from 3SAT to Backgammon KSKR is
    polynomial in space and time complexity
    with respect to the number of variables
    and the number of clauses in the 3SAT instance.
\end{lemma}

\begin{proof}
The length of the board is linear with respect to the
variables and clauses plus some constant buffer on either end.
Player 1 has at most twice the number of clauses
for each variable while Player 2 has at most a constant
number of pieces per variable and clause.
Only one piece is captured per reduction so the number
of moves is at most the product of the length
of the board and the number of pieces.

While it is potentially exponential with respect to the input,
the description of the dice rolls and opponent's move
may be stored in polynomial space due to their simplicity.
In the assignment stage,
the rolls and opponent moves are the same for
each variable gadget and can be stored in constant space
plus a pointer for the current index.
In the propagation stage,
the rolls and opponent moves are almost the same
for each variable gadget and clause gadget except
that the number of rolls necessary to move
between the variable and clause gadgets varies.
However, we can store the number of rolls by the index
in addition to constant space for the rules.
In the end game, Player 1 and Player 2 both
move with doubles if able and the rolls are repeated
until one player wins.
\end{proof}

\section{PSPACE-Hardness}\label{sec:pspace}

In this section, we show that 
determining whether a player
can win against an unknown opponent's strategy
and known dice rolls (USKR) is PSPACE-Hard.
We begin with formal definitions of Backgammon
USKR and the PSPACE-Complete problem we reduce from.

\begin{definition}[Backgammon USKR]
The input is a configuration on a generalized
backgammon board, an opponent's strategy which
is unknown to the player,
and known dice rolls.
The problem is to determine whether a player
can win the backgammon game from the backgammon
board against the opponent's unknown
strategy and with the specified dice rolls.
\end{definition}

An opponent's strategy is unknown if the player
does not know what the opponent will play given
a possible position and dice rolls
in the resulting game.
The opponent's strategy is not necessarily
deterministic; it can be adaptive or stochastic.

\begin{definition}[$\mathbf{G_{pos}}$ \cite{schaefer1978complexity}]
The input is a positive CNF formula
(without negations) on which two players
will play a game.
Player 1 and Player 2 alternate setting
exactly one variable of their choosing.
Once it has been set, a variable may not be set again.
Player 1 wins if and only if the formula evaluates
to True after all variables have been set.
The problem is to determine whether Player 1
can win.
\end{definition}

Given any $G_{pos}$ instance, we construct a backgammon
board configuration, an unknown opponent's strategy,
and known dice rolls so that the solution to Backgammon
USKR yields the solution to $G_{pos}$.
Since $G_{pos}$ is PSPACE-Complete 
\cite{schaefer1978complexity},
Backgammon USKR must be PSPACE-Hard.
We state the result formally below.

\begin{theorem}\label{thm:pspace}
Backgammon USKR is PSPACE-Hard.
\end{theorem}

\begin{proof}
The reduction from $G_{pos}$ to Backgammon USKR
closely follows the reduction from $3SAT$ to
Backgammon KSKR so we primarily
focus on the differences.
Assume we are given an arbitrary $G_{pos}$
instance with $n$ variables and $k$ clauses.
The key observation is that,
since the CNF is positive, Player 1
will always set variables to True while Player 2
will always set variables to False.
We can therefore equivalently think about the game
as Player 1 moving variables to a True position
while Player 2 blocks variables from becoming True.
Once all the variables have been set,
we propagate the choices to the clause gadgets
as we did in the 3SAT reduction.

The winning conditions also remain the same.
Player 1 wins if and only if they are 
able to cover the open piece in each clause.
We require the opponent's strategy to be
unknown so they can adversarially set variables.

\begin{figure}
    \centering
    \includegraphics[scale=.3]{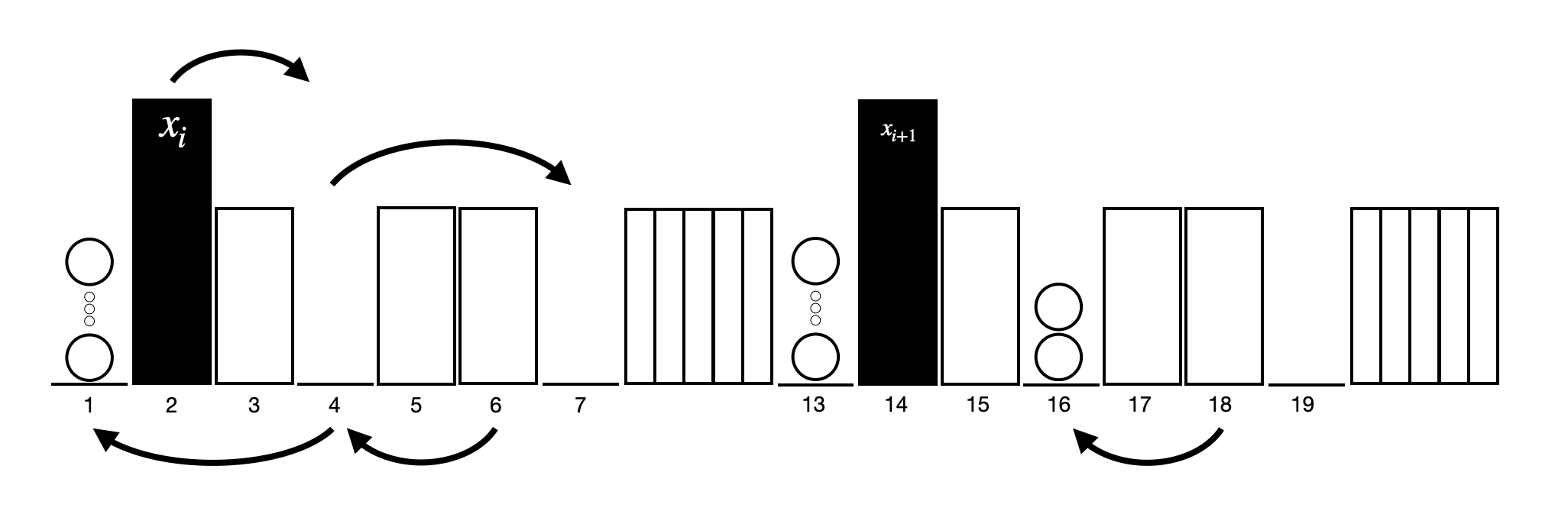}
    \caption{Reduction from $G_{pos}$: variable gadgets.
    Player 1 sets $x_i$ to True while Player 2
    has already set $x_{i+1}$ to False.}
    \label{fig:pspacehard}
\end{figure}

We now describe the variable gadgets in
Figure \ref{fig:pspacehard}.
In the 3SAT reduction, we needed a stack
of $x_i$ pieces and a stack of $\overline{x_i}$
pieces since Player 1 could set $x_i$ to
True or False.
Here, since the CNF is positive, Player 1
will only set variables to True and
so only require an $x_i$ stack.

At the beginning of Player 1's turn,
all unset variables are blocked by two
pieces on point 4.
Then, while Player 1 receives dummy rolls
of one, Player 2 moves all the blocking pieces
to 1.
Next, Player 1 receives a roll of two and three
and must choose which unset variable to set to True.
Once the variable is chosen,
Player 2 blocks all other unset variables by
moving two pieces from 6 to 4 while Player 1
again receives dummy rolls.
The remaining pieces for the chosen variable
are then moved from 2 to 4 to 7.

Player 2's turn is more simple.
They choose 
which variable to set to False and do so by blocking
16 from 18.
For the remainder of Player 1's turns,
the blocking pieces on 16 will not be moved.

After all the variables have been set
to True or False,
we move the pieces set to True to
the clauses they appear in.
We again move from $x_n$ to $x_1$,
removing the blocking pieces on 13 and then 1
as we go.
The process and clause gadgets are the same as in
the 3SAT reduction.

Player 1 wins in Backgammon USKR
if and only if the
positive CNF instance in $G_{pos}$ is
True after alternating setting variables
with Player 2.
Therefore the solution to Backgammon USKR
yields an answer to $G_{pos}$
and, by Lemma \ref{lemma:pspace},
$G_{pos}$ reduces in polynomial space
to Backgammon USKR.
\end{proof}

\begin{lemma}\label{lemma:pspace}
    The reduction from $G_{pos}$ to Backgammon
    USKR is polynomial in space complexity
    with respect to the number of clauses and variables
    in $G_{pos}$.
\end{lemma}

\begin{proof}
    As in the reduction from 3SAT to Backgammon KSKR,
    the size of the board is linear in the number of
    clauses and variables plus some constant buffer.
    By similar arguments, the size of the description
    is polynomial because it depends only on the stage
    of the reduction and the index of the current variable
    and clause gadgets.
\end{proof}

\section{EXPTIME-Hardness}\label{sec:exptime}

In this section, we show that determining
whether a player
can win against an unknown opponent strategy
and unknown
dice rolls (USUR) is EXPTIME-Hard.
We begin with formal definitions of Backgammon
USUR and the EXPTIME-Complete problem we reduce from.

\begin{definition}[Backgammon USUR]
The input is a configuration on a generalized
backgammon board.
The opponent's strategy and dice rolls are
unknown to the player.
The problem is to determine whether a player
can win the backgammon game from the backgammon
configuration against the unknown strategy and
dice rolls.
\end{definition}

\begin{definition}[$\mathbf{G_6}$ \cite{stockmeyer1979provably}]
The input is a CNF formula on sets of variables
$X$ and $Y$ and an initial assignment of the variables.
Player 1 and Player 2 alternate changing
at most one variable.
Player 1 may only change variables in $X$
while Player 2 may only change variables in $Y$.
Player 1 wins if the formula ever becomes true.
The problem is to determine whether Player 1
can win.
\end{definition}
Given any $G_6$ instance, we construct
a backgammon board configuration and exhibit
an opponent's strategy and dice rolls such that
the solution to Backgammon USUR yields the
solution to $G_6$.
Since $G_6$ is EXPTIME-Complete
\cite{stockmeyer1979provably},
Backgammon USUR must be EXPTIME-Hard.

\begin{theorem}
Backgammon USUR is EXPTIME-Hard.
\end{theorem}

\begin{proof}
The proof consists of a reduction
from $G_6$ to Backgammon USUR.
Assume we are given a CNF formula
with $n_x$ variables $X$, $n_y$ variables
$Y$, $k$ clauses,
and an initial assignment to $X$ and $Y$.
First, Player 1 and Player 2 take turns
changing variables in their respective
sets $X$ and $Y$.
Then, once Player 1 gives the signal,
the game progresses to a board state where
Player 1 wins if and only if the CNF formula
is True on the current assignment.

Player 1 changes variable $x_i$ by moving
a signal piece corresponding to $x_i$.
Then, with Player 2's help, we feed Player 1
dice rolls that update the clause gadgets
that contain $x_i$.
We require that the dice rolls adaptively
respond to Player 1 and that Player 2 can
adversarially set variables in $Y$.

\begin{figure}
    \centering
    \includegraphics[scale=.18, page=9]{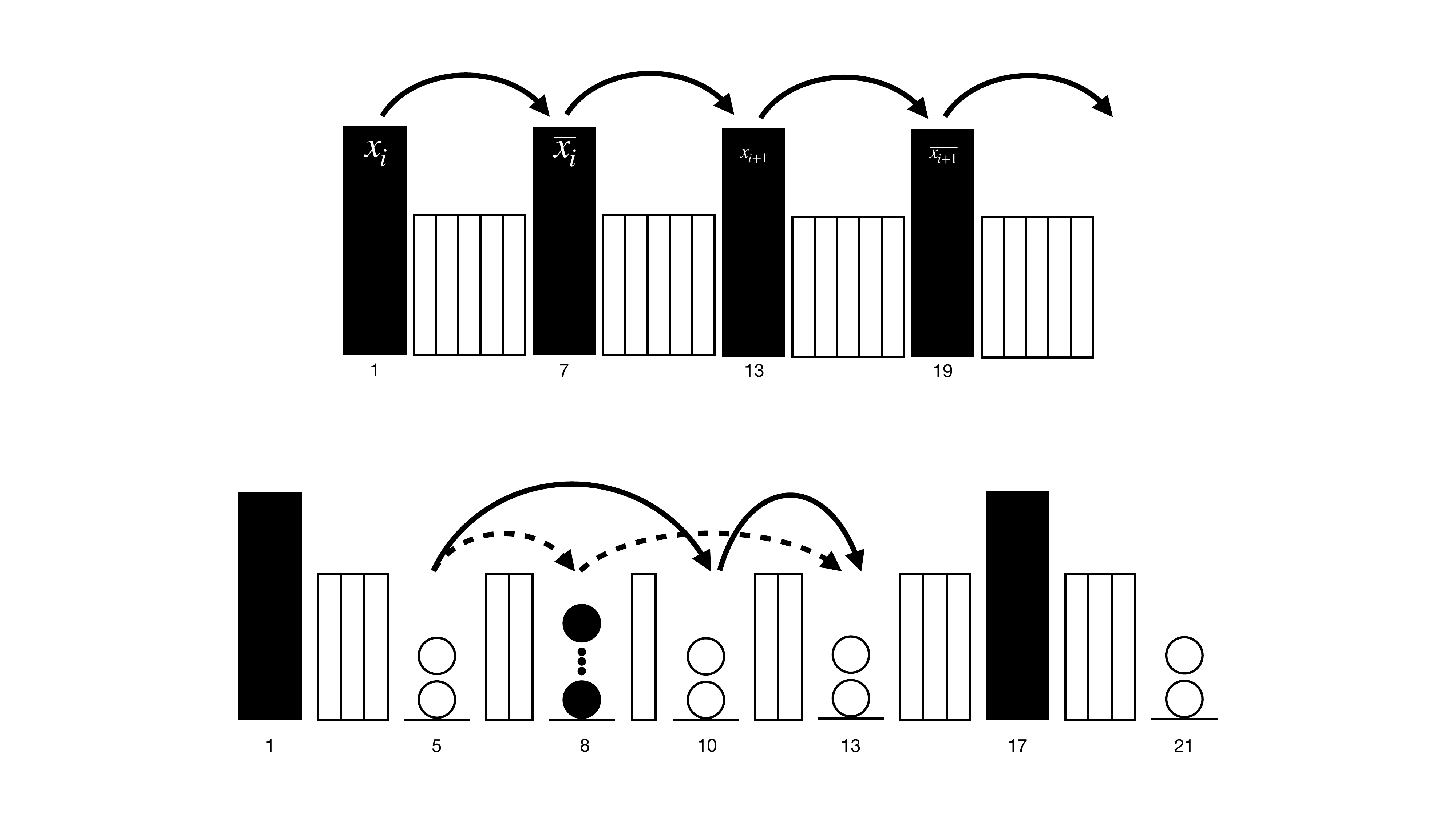}
    \caption{Reduction from $G_6$: variable gadget (top) and clause gadget (bottom).}
    \label{fig:exptimehard}
\end{figure}

We next describe the gadgets in Figure 
\ref{fig:exptimehard}.
The variable gadgets consist of stacks of
pieces corresponding to $x_i$ and $\overline{x_i}$
for $i \in [n_x]$.
On their turn, Player 1 changes a variable
by using their six to move the appropriate piece.
For example, Player 1 can change $x_i$ to False
by moving a piece on point 7 to 13 as shown
at the top of Figure \ref{fig:exptimehard}.
If $x_i$ is already False, Player 1 has effectively
skipped their turn (which is an acceptable
move in $G_6$).

Once Player 1 changes a variable,
Player 2 and the dice rolls work together to update
the appropriate clauses.
The key insight of a clause is that it is
True if at least one variable in the clause
is True in either $X$ or $Y$.
We represent this on the backgammon board as
shown at the bottom of Figure \ref{fig:exptimehard}.
Point 8 is empty or contains Player 1's
pieces if at least one of the $X$ variables
in the clause is True and 10 is empty
if at least one of the $Y$ variables in
the clause is True.
Therefore Player 1 can progress a piece from 5 to 13
on rolls three and five if and only if
either a variable in $X$
or $Y$ in the clause is True.

We update the clause when Player 1 sets
a variable in one of two ways:
If the variable is True in the clause,
we move two pieces from 1
to 5 to 8.
If the variable is False in the clause,
we move two pieces from 8 to 13 to 17.
If 8 becomes empty, we move two pieces from
a white stack to 8 in order to block Player 1
from unfairly using it to bypass the clause.
By using blocking pieces on 5 and 13 we ensure
the correct clause is modified.

We update the clause when Player 2 sets
a variable in one of two ways:
If the variable is True in the clause and
all other $Y$ variables in the clause are True,
we move two pieces from 10 to another white stack.
If the variable is False in the clause and
all other $Y$ variables in the clause are True,
we move two pieces from another white stack to 10.
In all other cases, 10 should remain in its
current `open' or `closed' position.

Notice that the process of changing variables
could continue indefinitely.
We make sure that we do not run out of pieces by using
the capture rule.
If Player 1 needs more pieces in the variable or
clause gadgets, we feed them rolls to move excess
pieces through the board towards their home where
Player 2 will capture them one by one.
We perform an analogous process if
Player 2 needs more pieces.

The variable changing process ends when,
instead of moving a variable piece,
Player 1 moves a specified signal
at the end of the variable gadgets.
Then Player 1 will receive enough six and 4-3-5-4 rolls
to move all of their pieces home while Player 2
makes slow progress.
If all of the clauses are True, Player 1 can
successfully get all their pieces home and win
the game.
Otherwise, they will be blocked at a False clause
and Player 2 will continue 
their slow progress until all white pieces
except for those in the False clause remain.
We will then feed small rolls to Player 1 and 
large rolls to Player 2 so Player 2 can capitalize
on their advantage and win.

We have therefore simulated the $G_6$ instance
and Player 1 can win Backgammon USUR if and only
if they can win the corresponding $G_6$ game.
As before, the reduction space is polynomial 
in the $G_6$ input size
since there are a constant number of pieces
and points for every clause and variable.
\end{proof}

Notice that Backgammon USUR is not obviously
EXPTIME-Complete because the game can progress
indefinitely as a result of the capture rule.

\section{Conclusion}

We show that deciding whether a player can win
a backgammon game under different settings of
known or unknown opponent strategies and dice
rolls is NP-Hard, PSPACE-Hard, and EXPTIME-Hard.
It would seem that our results show backgammon
is hard even when it is a one-player game.
However, in the settings for our PSPACE-Hardness
and EXPTIME-Hardness results, the second player
is hidden in the unknown nature of the opponent's
strategy and dice rolls.

Despite the popularity of backgammon and
academic interest in the computational complexity
of games, to the best of our knowledge
our work is the first to address
the complexity of backgammon.
One possible explanation is the apparent
ambiguity in generalizing backgammon.
We contend, however, that the backgammon
generalization we use is as natural as those for
other games like checkers or chess.
Another explanation is the difficulty
in forcing backgammon moves as needed for a reduction.

Interestingly, it is not clear that the
problems we consider are in EXPTIME because
backgammon games can theoretically continue
indefinitely.
One natural follow up question to our work
is what complexity class contains these
backgammon problems.

\bibliographystyle{splncs04}
\bibliography{bibliography}
\end{document}